\documentclass[10pt]{article}

\usepackage{amsfonts,amssymb,mathtools}
\usepackage{amssymb}
\usepackage{amsmath}
\usepackage{graphicx}
\usepackage[english]{babel}
\usepackage[latin1]{inputenc}
\usepackage{verbatim}
\usepackage{enumerate}
\usepackage{tikz}
\usetikzlibrary{backgrounds,calc}
\usepackage{float}
\usepackage{color}
\usepackage{acronym}
\usepackage{algorithm}
\usepackage{algorithmic}
\newfloat{algorithm}{H}{loa}
\floatname{algorithm}{Algorithm}
\newcounter{alg}
\usepackage{multirow}

\DeclareMathOperator*{\argmin}{arg\min}

\newtheorem{theorem}{Theorem}[section]

\newtheorem{claim}[theorem]{Claim}

\newtheorem{definition}{Definition}[section]

\newtheorem{example}{Example}[section]

                      {}

                      {}

\def\squarebox#1{\hbox to #1{\hfill\vbox to #1{\vfill}}}

\def\eod{\vrule height 6pt width 5pt depth 0pt}
\newenvironment{proof}{\noindent {\bf Proof:} \hspace{.677em}}
                      {\hspace*{\fill}{\eod}}

\newcommand\Eq{\mathcal{E}}

\newcommand\G{\mathcal{G}}

\newcommand{\tuple}[1]{\langle #1  \rangle}
\allowdisplaybreaks

\begin{document}

\title{Scheduling Games with Machine-Dependent Priority Lists
\footnote{A preliminary version appeared in the proceedings of the $15$th Conference on Web and Internet Economics (WINE), December 2019.}}

\author{Vipin Ravindran Vijayalakshmi\thanks{Chair of Management Science, RWTH Aachen, Germany, E-mail: vipin.rv@oms.rwth-aachen.de. This work is supported by the German research council (DFG) Research Training Group 2236 UnRAVeL.} \and Marc Schr{\"o}der \thanks{School of Business and Economics, Maastricht University, Netherlands, E-mail: m.schroder@maastrichtuniversity.nl} \and
Tami Tamir\thanks{School of Computer Science, The Interdisciplinary Center, Israel. E-mail:tami@idc.ac.il. This work is supported by The Israel Science Foundation (ISF). Grant No. 1036/17.}}

\date{}
\maketitle

\begin{abstract}
We consider a scheduling game on parallel related machines, in which jobs try to minimize their completion time by choosing a machine to be processed on. Each machine uses an individual priority list to decide on the order according to which the jobs on the machine are processed. 
We prove that it is NP-hard to decide if a pure Nash equilibrium exists and characterize four classes of instances in which a pure Nash equilibrium is guaranteed to exist. For each of these classes, we give an algorithm that computes a Nash equilibrium, we prove that best-response dynamics converge to a Nash equilibrium, and we bound the inefficiency of Nash equilibria with respect to the makespan of the schedule and the sum of completion times. In addition, we show that although a pure Nash equilibrium is guaranteed to exist in instances with identical machines,
it is NP-hard to approximate the best Nash equilibrium with respect to both objectives. 
\end{abstract}

\section{Introduction}
Scheduling problems have traditionally been studied from a centralized point of view in which the goal is to find an assignment of jobs to machines so as to minimize some global objective function. Two of the classical results are that Smith's rule, i.e., schedule jobs in decreasing order according to their ratio of weight over processing time, is optimal for single machine scheduling with the sum of weighted completion times objective \cite{smith1956various}, and list scheduling, i.e., greedily assign the job with the highest priority to a free machine, yields a $2$-approximation for identical machines with the minimum makespan objective \cite{Gra66}. Many modern systems provide service to multiple strategic users, whose individual payoff is affected by the decisions made by others. As a result, non-cooperative game theory has become an essential tool in the analysis of job-scheduling applications. The jobs are controlled by selfish users who independently choose which resources to use. The resulting job-scheduling games have by now been widely studied and many results regarding the efficiency of equilibria in different settings are known.

A particular focus has been placed on finding coordination mechanisms \cite{christodoulou2004coordination}, i.e., local scheduling policies, that induce a good system performance. In these works it is common to assume that ties are broken in a consistent manner (see, e.g., Immorlica et al.~\cite{immorlica2009coordination}), or that there are no ties at all (see, e.g., Cole et al.~\cite{cole2015decentralized}). In practice, there is no real justification for this assumption, except that it avoids subtle difficulties in the analysis. In this paper we relax this restrictive assumption and consider the more general setting in which machines have arbitrary individual priority lists. That is, each machine schedules those jobs that have chosen it according to its priority list. The priority lists are publicly known to the jobs.  

In this paper we analyze the effect of having  machine-dependent priority lists on the corresponding job-scheduling game. We study the existence of pure Nash equilibria (NE), the complexity of identifying whether an NE profile exists, the complexity of calculating an NE, in particular a good one, and the equilibrium inefficiency. 



\subsection{The Model}
An instance of a {\em scheduling game with machine-dependent priority lists} is given by a tuple $G=\tuple{N,M,(p_i)_{i \in N},(s_j)_{j \in M},(\pi_j)_{j \in M}}$, where $N$ is a finite set of $n\geq 1$ jobs, $M$ is a finite set of $m\geq 1$ machines, $p_i\in\mathbb{R}_+$ is the processing time of job $i\in N$, $s_j\in\mathbb{R}_+$ denotes the speed of the machine $j\in M$, and $\pi_j:N \rightarrow \{1,\ldots,n\}$ is the priority list of machine $j\in M$. 

A strategy profile $\sigma=(\sigma_i)_{i\in N}\in M^N$ assigns a machine $\sigma_i\in M$ to every job $i\in N$. Given a strategy profile $\sigma$, the jobs are processed according to their order in the machines' priority lists. 
The set of jobs that delay $i\in N$ in $\sigma$ is denoted by $B_{i}(\sigma)= \{i' \in N| \sigma_{i'}=\sigma_i \wedge \pi_{\sigma_i}(i') \leq \pi_{\sigma_i}(i)\}$. Note that job $i$ itself also belongs to $B_{i}(\sigma)$.  
Let $p_i(\sigma)=\sum\limits_{i'\in B_{i}(\sigma)}p_{i'}$. The cost of job $i\in N$ is equal to its completion time in $\sigma$, given by
$C_i(\sigma)=p_i(\sigma)/s_{\sigma_i}.$

Each job chooses a strategy so as to minimize its costs. A strategy profile $\sigma\in \Sigma$ is a {\em pure Nash equilibrium (NE)} if for all $i\in N$ and all $\sigma_i^{\prime}\in \Sigma_i$, 
we have that $C_i(\sigma)\leq C_i(\sigma_i^{\prime},\sigma_{-i})$. Let $\Eq(G)$ denote the set of Nash equilibria for a given instance $G$. We would like to remark that $\Eq(G)$ may be empty.

For a strategy profile $\sigma$, let $C(\sigma)$ denote the cost of $\sigma$. The cost is defined with respect to some objective, e.g., the makespan, i.e., $C_{max}(\sigma):=\max_{i\in N} C_i(\sigma)$, or the sum of completion times, i.e., $\sum_{i \in N}C_i(\sigma)$. It is well known that decentralized decision-making may lead to sub-optimal solutions from the
point of view of the society as a whole. For a game $G$, let $P(G)$ be the set of feasible profiles of $G$. We denote by $OPT(G)$ the cost of a social optimal solution, i.e., $OPT(G)=\min_{\sigma \in P(G)} C(\sigma)$. We quantify the
inefficiency incurred due to self-interested behavior according to the
\emph{price of anarchy} (PoA) \cite{Koutsoupias:1999:WE:1764891.1764944}, and \emph{price of
stability} (PoS) \cite{AD+08}. The PoA is the worst-case inefficiency of a pure Nash equilibrium, while the PoS measures the best-case inefficiency of a pure Nash equilibrium. 

\begin{definition}
\label{def:ineff}
Let $\G$ be a family of games, and let $G$ be a game in $\mathcal{G}$.
Let $\Eq(G)$ be the set of pure Nash equilibria of the game $G$. Assume that $\Eq(G) \neq \emptyset$.
\begin{itemize}
\item The {\em price of anarchy} of $G$ is the ratio between the
\emph{maximum} cost of an NE and the social optimum of
$G$, i.e.,
$\mbox{PoA}(G) = \max\limits_{\sigma\in \Eq(G)} C(\sigma)/OPT(G)$.
The {\em price of anarchy} of $\mathcal{G}$
is $\mbox{PoA}(\mathcal{G}) = sup_{ G\in \mathcal{G}}\mbox{PoA}(G)$.
\item
The {\em price of stability} of $G$ is the ratio between the
\emph{minimum} cost of an NE and the social optimum of
$G$, i.e.,
$\mbox{PoS}(G) = \min\limits_{\sigma\in \Eq(G)} C(\sigma)/OPT(G)$.
The {\em price of stability} of $\mathcal{G}$ is
$\mbox{PoS}(\mathcal{G}) = sup_{ G\in \mathcal{G}}\mbox{PoS}(G)$.
\end{itemize}
\end{definition}

\subsection{Our Contribution}
We first show that a pure Nash equilibrium in general need not exist, and use this to show that it is NP-complete to decide whether a particular game has a pure Nash equilibrium. 
We then provide a characterization of instances in which a pure Nash equilibrium is guaranteed to exist. Specifically, existence is guaranteed if the game belongs to at least one of the following four classes: $\G_1:$ all jobs have unit processing time, $\G_2:$ there are two machines, $\G_3:$ all machines have the same speed, and $\G_4:$ all machines have the same priority list. For all four of these classes, there is a polynomial time algorithm that computes a Nash equilibrium. In fact, for all four classes we prove that better-response dynamics converge to a Nash equilibrium. This characterization is tight in a sense that our inexistence example disobeys it in a minimal way: it describes a game on three machines, two of them having the same speed and the same priority list. We also show that the result for $\G_2$ cannot be extended for games with two unrelated machines. Another characterization we consider is the number of jobs in the instance. We present a game of $4$ jobs that has no pure NE, and show that every game of $3$ jobs admits an NE.

We analyze the efficiency of Nash equilibria by means of two different measures of efficiency: the makespan, i.e., the maximum completion time of a job, and the sum of completion times. For all four classes of games with a guaranteed pure Nash equilibrium, we provide tight bounds for the price of anarchy and the price of stability with respect to both measures. Our results are summarized in Table~\ref{tab:results}.

\begin{table*}[htbp]
 \begin{center}
  \begin{tabular}{|c||c|c|}
  \hline
  \multirow{2}{*}{Instance class $\char`\\$ Objective} & Makespan & Sum of Comp. Times\\
  & PoA/PoS & PoA/PoS \\
  \hline\hline
  $\G_1:$ Unit jobs  & $1$ & $1$ \\
  \hline
  $\G_2:$ Two machines & $(\sqrt{5}+1)/2$& $\Theta(n)$ \\
  \hline
  $\G_3:$ Identical machines & $~~2-1/m~~$ & $\Theta(n/m)$ \\
    \hline
  $\G_4:$ Global priority list & $\Theta(\log{m})$ & $\Theta(n)$ \\
  \hline
  \end{tabular}
 \end{center}
 \caption{Our results for the equilibrium inefficiency.}
 \label{tab:results}
 \end{table*}
 
 $(i)$ If jobs have unit processing times, we show that the price of anarchy is equal to $1$, which means that selfish behavior is optimal. $(ii)$ For two machines with speed $1$ and $s\leq 1$ respectively, we prove that the PoA and the PoS are at most $s+1$ if $s\le\frac{\sqrt{5}-1}{2}$, and $\frac {s+2}{s+1}$ if $s\ge\frac{\sqrt{5}-1}{2}$.
Moreover, our analysis is tight for all $s \le 1$. The maximal inefficiency, listed in Table \ref{tab:results}, is achieved for $s=\frac{\sqrt{5}-1}{2}$. In case the sum of completion times is considered as an objective, the price of anarchy can grow linearly in the number of jobs. $(iii)$ If machines have identical speeds, but potentially different priority lists, the price of anarchy with respect to the makespan is equal to $2-1/m$. The upper bound follows because every Nash equilibrium can be seen as an outcome of Graham's List-Scheduling algorithm. This generalizes a similar result by Immorlica et al.~\cite{immorlica2009coordination} for priorities based on shortest processing times first. The lower bound example shows the bound is tight, even with respect to the price of stability. For the sum of completion times objective, we show that the price of anarchy is at most $O(n/m)$, and provide a lower bound example for which the price of stability grows in the order of $O(n/m)$. $(iv)$ If there is a global priority list, and machines have arbitrary speeds, we show that the $\Theta(\log{m})$-approximation of List-Scheduling carry over for the makespan inefficiency, and the results for two machines carry over for the sum of completion times.

We conclude with results regarding the complexity of calculating a good NE. While a simple greedy algorithm can be used to compute an NE for an instance with identical machines (the class ${\cal G}_3$), we show that it is NP-hard to compute an NE schedule that approximates the best NE of a game in this class. Specifically, it is NP-hard to approximate the best NE with respect to the minimum makespan within a factor of $2-1/m-\epsilon$ for all $\epsilon>0$, and it is NP-hard to approximate the best NE with respect to the sum of completion times within a factor of $r$ for any constant $r>1$.


\subsection{Related Work}

Scheduling games were initially studied in the setting in which each machine processes its jobs in parallel so that the completion time of each job is equal to the makespan of the machine. The goal of these papers was to characterize the inefficiency of selfish behavior as measured by the price of anarchy \cite{Koutsoupias:1999:WE:1764891.1764944}. Most attention has been given to the makespan as a measure of efficiency. Czumaj and V\"ocking~\cite{czumaj2007tight} gave tight bounds on the price of anarchy for related machines, whereas Awerbuch et al.~\cite{awerbuch2006tradeoffs} and Gairing et al.~\cite{gairing2010computing} provided tight bounds for restricted machine settings. We refer to V\"ocking~\cite{Voc07} for an overview. These tight bounds grow with the number of machines and that is why Christodoulou et al.~\cite{christodoulou2004coordination} introduced the idea of using coordination mechanisms, i.e., local scheduling policies, to improve the price of anarchy. They studied the price of anarchy with priority lists based on longest processing times first. Immorlica et al.~\cite{immorlica2009coordination} generalized their results and studied several different scheduling policies, among which longest and shortest processing times first, in multiple scheduling settings. Both these two policies guarantee the existence of a pure Nash equilibrium in the related machine setting. These results are a special case of our result, as we prove the existence of a pure Nash equilibrium if there is a global priority list. For shortest processing times first, a pure Nash equilibrium is also guaranteed in the unrelated machines setting. Here, the set of Nash equilibria corresponds to the set of solutions of the Ibarra-Kim algorithm. A result that is also proven in Heydenreich et al.~\cite{heydenreich2007games}. Other (in)existence results are D\"urr and Nguyen~\cite{durr2009non}, who proved that a Nash equilibrium exists for two machines with a random order and balanced jobs, Azar et al.~\cite{azar2008almost}, who showed that for unrelated machines with priorities based on the ratio of a job's processing time to its faster processing time a Nash equilibrium need not exist, Lu and Yu~\cite{lu2012worst}, who proved that group-makespan mechanisms guarantees the existence of a Nash equilibrium, and Kollias~\cite{kollias2013nonpreemptive}, who showed that non-preemptive coordination mechanisms need not induce a pure Nash equilibrium.  

For the sum of weighted completion times, Correa and Queyranne~\cite{correa2012efficiency} proved a tight upper bound of 4 for restricted related machines with priority lists derived from Smith's rule. Cole et al.~\cite{cole2015decentralized} generalized the bound of 4 to unrelated machines with Smith's rule and proposed better scheduling policies. Hoeksma and Uetz~\cite{hoeksma2019price} gave a tighter bound for the more restricted setting in which jobs have unit weights and machines are related. Caragiannis et al.~\cite{caragiannis2017coordination} proposed a framework that uses price of anarchy results of Nash equilibria in scheduling games to come up with combinatorial approximation algorithms for the centralized problem.


Ackermann et al.~\cite{Ackermann:2007:UAC:1781894.1781903} were the first to study a congestion game with priorities. They proposed a model in which users with higher priority on a resource displace users with lower priorities such that the latter incur infinite cost. Closer to ours is Farzad et al.~\cite{DBLP:journals/cjtcs/FarzadOV08}, who studied priority based selfish routing for non-atomic and atomic users and analyzed the inefficiency of equilibria. Recently, Bil\'o and Vinci~\cite{bilo2020congestion} studied a congestion game with a global priority classes that can contain multiple jobs and characterize the price of anarchy as a function of the number of classes. Gourv{\`e}s et al.~\cite{Gourves:2015:CGC:2839125.2839192} studied capacitated congestion games to characterize the existence of pure Nash equilibria and computation of an equilibrium when they exist. Piliouras et al.~\cite{piliouras2016risk} assumed that the priority lists are unknown to the players a priori and consider different risk attitudes towards having a uniform at random ordering.

\section{Equilibrium Existence and Computation}
\label{sec:noNEinSW}
In this section we give a precise characterization of scheduling game instances that are guaranteed to have an NE. The conditions that we provide are sufficient but not necessary. A natural question is to decide whether a given game instance that does not fulfill any of the conditions has an NE. We show that answering this question is an NP-complete problem.

We first show that an NE may not exist, even with only three machines, two of which have the same speed and the same priority list.

\begin{example}
\label{example:noNEsched}
{\em
Consider the game $G^*$ with $5$ jobs,  $N = \{a,b,c,d,e\}$, and three machines, $M = \{M_1, M_2, M_3\}$, with $\pi_1=(a,b,c,d,e)$, and $\pi_2= \pi_3= (e,d,b,c,a)$.  The first machine has speed $s_1=1$ while the two other machines have speed $s_2=s_3=1/2$.
The job processing times are $p_a = 5, p_b = 4, p_c = 4.5, p_d = 9.25$, and $p_e=2$.

Job $a$ is clearly on $M_1$ in every NE. Therefore job $e$ is not on $M_1$ in an NE, as job $e$ is first on $M_2$ or $M_3$. Since these two machines have the same priority list and the same speed, we can assume w.l.o.g., that if an NE exists, then there exists an NE in which job $e$ is on $M_3$. We distinguish two different cases for job $d$, as given that $e$ is on $M_3$, $d$ prefers $M_2$ over $M_3$.
\begin{enumerate}
    \item Job $d$ is on $M_1$. Then as job $b$ has the highest remaining priority among $b$ and $c$ on all machines, job $b$ picks the machine with the lowest completion time, which is $M_2$, and job $c$ lastly is then on $M_1$. As a result, $d$ prefers $M_2$ (since $18.5< 18.75$) over $M_1$.
    \item Job $d$ is on $M_2$. Then as job $b$ has the highest remaining priority among $b$ and $c$ on all machines, job $b$ picks the machine with the lowest completion time, which is $M_1$, and job $c$ lastly is then on $M_3$. As a result, $d$ prefers $M_1$ (since $18.25< 18.5$) over $M_2$.
\end{enumerate}
Thus, the game $G^*$ has no pure Nash equilibrium.

}
\end{example}

We can use the above example to show that deciding whether a game instance has an NE is NP-complete by using a reduction from 3-bounded 3-dimensional matching. 
\begin{theorem}
\label{thm:hard:NEexistsSchedule}
Given an instance of a scheduling game, it is NP-complete to decide whether the game has an NE. 
\end{theorem}
\begin{proof}
Given a game and a profile $\sigma$, verifying whether $\sigma$ is an NE can be done by checking for every job whether its current assignment is also its best-response, therefore the problem is in NP.

The hardness proof is by a reduction from 
$3$-bounded $3$-dimensional matching ($3$DM-$3$).
The input to the $3$DM-$3$ problem is a set of triplets $T \subseteq X \times Y \times Z$,
where $|T|\ge n$ and $|X|=|Y|=|Z|=n$. The number of occurrences of every element of $X
\cup Y \cup Z$ in $T$ is at most $3$. The goal is to decide whether $T$ has a $3$-dimensional matching of size $n$,  i.e., 
there exists a subset $T' \subseteq T$, such that $|T'|=n$, and every element in $X \cup Y \cup Z$ appears exactly once in $T'$.
$3$DM-$3$ is known to be NP-hard \cite{Kann91}.

Given an instance $T$ of $3$DM-$3$ matching and $\epsilon >0$, we construct the following scheduling game, $G_T$.
The set of jobs consists of:
\begin{enumerate}
    \item The $5$ jobs $\{a,b,c,d,e\}$ from the game $G^*$ in Example \ref{example:noNEsched}.
    \item A single dummy job, $f$, with processing time $2$.
    \item A set $D$ of $|T|-n$ dummy jobs with processing time $3$.
    \item A set $U$ of $|T|+1$ dummy jobs with processing time $20$.
    \item $3n$ jobs with processing time $1$ - one for each element in $X \cup Y \cup Z$.
\end{enumerate}
There are $m=|T|+4$ machines, $M_1,M_2,\ldots,M_{|T|+4}$. 
All the machines except for $M_2$ and $M_3$ have speed $s_j=1$. For $M_2$ and $M_3$, $s_2=s_3=1/2$. 

The heart of the reduction lies in determining the priority lists. The first three machines will mimic the no-NE game $G^*$ from Example \ref{example:noNEsched}. Note that if job $e$ is missing from $G^*$ then there exists an NE of $\{a,b,c,d\}$ on $M_1,M_2,M_3$. The idea is that if a $3$DM-$3$ matching exists, then job $e$ would prefer $M_4$ and leave the first three machines for $\{a,b,c,d\}$. However, if there is no $3$DM-$3$, then some job originated from the elements in $X \cup Y \cup Z$ will precede job $e$ on $M_4$, and $e$'s best-response would be on $M_2$ or $M_3$ - where it is guaranteed to have completion time $4$, and the no-NE game $G^*$ would come to life. The dummy jobs in $U$ are long enough to guarantee that each of the jobs $\{a,b,c,d\}$ prefers the first three machines over the last $|T|+1$ machines.

The priority lists are defined as follows. When the list includes a set, it means that the set elements appear in arbitrary order. For the first machine, $\pi_1=(a,b,c,d,e,f,U,X,Y,Z,D)$.
For the second and third machines, $\pi_2=\pi_3 =(e,d,b,c,a,f,U,X,Y,Z,D)$.
For the fourth machine, we have priority list $\pi_4 =(f,X,Y,Z,e,U,D,a,b,c,d)$.
The remaining $|T|$ machines are {\em triplet-machines}. 
For every triplet $t=(x_i,y_j,z_k) \in T$, the priority list of the triplet-machine corresponding to $t$ is $(D,x_i,y_j,z_k,U,f, X\setminus\{x_j\},  Y\setminus\{y_j\}, Z\setminus\{z_j\}, a,b,c,d,e)$.

Observe that in any NE, the dummy job $f$ with processing time $2$ is assigned as the first jobs on $M_4$. Also, the dummy jobs in $D$ have the highest priority on the triplet-machines, thus, in every NE, there are $|D|=|T|-n$ triplet-machines on which the first job is from $D$. Finally, it is easy to see that in every NE there is exactly one dummy job from $U$ on each of the last $|T|+1$ machines.

 Figure~\ref{fig:hardNE2} provides an example for $n=2$ and $|T|=3$.

\begin{figure*}[ht]
\begin{center}
\includegraphics[height=5.5cm]{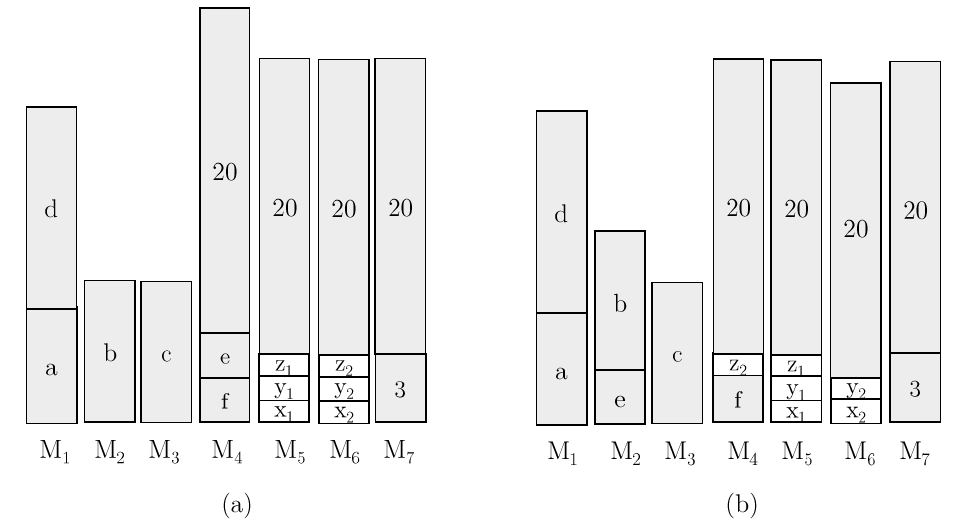}
\end{center}
\caption{(a) Let $T=\{(x_1,y_1,z_1),(x_2,y_2,z_2),(x_1,y_2,z_2)\}$.
A matching of size $2$ exists. Job $e$ is assigned on $M_4$, an NE exists. (b) Let $T=\{(x_1,y_1,z_1),(x_2,y_2,z_1),(x_1,y_2,z_2)\}$.
A matching of size $2$ does not exist. Job $e$ is not assigned on $M_4$ and the no-NE game $G^*$ is induced on the first three machines.}
\label{fig:hardNE2}
\end{figure*}

In order to complete the proof we prove the following two claims that relate the existence of a NE in the game $G_T$ to the existence of a perfect matching in the 3DM-3 instance $T$.   We first show that if the 3DM-3 instance has a perfect matching, then the game induced due to our construction has a pure Nash equilibrium.
\begin{claim}\label{clm:1:hard:NEexistsSchedule}
If a $3$D-matching of size $n$ exists in $T$, then the game $G_T$ has an NE. 
\end{claim}
\begin{proof}
Let $T' \subseteq T$ be a matching of size $n$. Assign the jobs of $X \cup Y \cup Z$ on the triplet-machines corresponding to $T'$ and the jobs of $D$ on the remaining triplet-machines. Assign $f$ and $e$ on $M_4$. Also, assign a single job from $U$ on all but the first $3$ machines. We are left with the jobs $a,b,c,d$ that are assigned on the first three machines: $a$ and $d$ on $M_1$, $b$ on $M_2$ and $c$ on $M_3$. It is easy to verify that the resulting assignment is an NE. The crucial observation is that all the jobs originated from $X \cup Y \cup Z$ completes at time at most $3$, and therefore have no incentive to select $M_4$. Thus, job $e$ completes at time $4$ on $M_4$ and therefore, has no incentive to join the no-NE game on the first three machines.
\end{proof}

The next claim shows that if the 3DM-3 instance does not have a perfect matching, then as a consequence of our construction, the no-NE subgame $G^*$ is triggered, and $G_T$ has no NE.
\begin{claim}\label{clm:2:hard:NEexistsSchedule}
If a $3$D-matching of size $n$ does not exist, then the game $G_T$ has no NE.
\end{claim}
\begin{proof}
Since a matching does not exist, at least one job from $X \cup Y \cup Z$, is not assigned on its triplet machine, and thus prefers $M_4$, where its completion time is $3$. Thus, job $e$ prefers to be first on $M_2$ or $M_3$, where its completion time is $4$. The long dummy jobs guarantee that machines $M_1,M_2$ and $M_3$ attracts exactly the $5$ jobs $\{a,b,c,d,e\}$ and the no-NE game $G^*$ is played on the first three machines. 
\end{proof}

The proof of Theorem~\ref{thm:hard:NEexistsSchedule} then immediately follows from claims~\ref{clm:1:hard:NEexistsSchedule} and \ref{clm:2:hard:NEexistsSchedule}.
\end{proof}


\vskip6pt
Our next results are positive. We introduce four classes of games for which an NE is guaranteed to exist. This characterization is tight in a sense that our inexistence example disobeys it in a minimal way. For classes $\G_3$ ($s_j=1$ for all $j\in M$) and $\G_4$ ($\pi_j=\pi$ for all $j\in M$), a simply greedy algorithm shows that an NE always exists. We refer to Correa and Queyranne~\cite{correa2012efficiency}, and Farzad et al.~\cite{DBLP:journals/cjtcs/FarzadOV08}, respectively, for a formal proof.

The following algorithm computes an NE for instances in the class $\G_1$, that is, $p_i=1$ for all $i \in N$. It assigns the jobs greedily, where in each step, a job is added on a machine on which the cost of a next job is minimized.

\begin{algorithm}[H]
\caption{ Calculating an NE of unit jobs on related machines} \label{NEUWjobs}
\begin{algorithmic}[1]
\STATE Let $\ell_j$ denote the number of jobs assigned on machine $j$. Initially, $\ell_j=0$ for all $1 \le j \le m$.
\REPEAT 
\STATE Let $j^{\star} = \argmin_j ~~ (\ell_j+1)/s_j$.
\STATE Assign on machine $j^{\star}$ the first unassigned job on its priority list.
\STATE $\ell_{j^{\star}} = \ell_{j^{\star}}+1$.
\UNTIL {all jobs are scheduled}
\end{algorithmic}
\end{algorithm}

\begin{theorem}
\label{thm:unitjobs}
If $p_i=1$ for all jobs $i \in N$, then Algorithm \ref{NEUWjobs} calculates an NE.
\end{theorem}
 \begin{proof}
 Let $\sigma^{\star}$ be the schedule produced by Algorithm \ref{NEUWjobs}. We show that $\sigma^{\star}$ is an NE.   Note that the jobs are assigned one after the other according to their completion time in $\sigma^{\star}$. That is, if $j_1$ is assigned before $j_2$ then $C_{j_1}(\sigma^{\star}) \le C_{j_2}(\sigma^{\star})$. 
 Assume by contradiction that $\sigma^{\star}$ is not an NE, and let $i$ be a job that can migrate from its current machine to machine $j$ and reduce its completion time. Assume that if it migrates, then $i$ would be assigned as the $k$-th job on machine $j$. This contradicts  the choice of the algorithm when the $k$-th job on machine $i$ is assigned - since $j$ should have been selected. If no job is $k$-th on machine $i$, then we get a contradiction to the assignment of $i$.
 \end{proof}

~

The following algorithm produces an NE for instances in the class $\G_2$, that is, $m=2$.

\begin{algorithm}[H]
\caption{ Calculating an NE schedule on two related machines} \label{NE2machines}
\begin{algorithmic}[1]
\STATE Assign all the jobs on $M_1$ (fast machine) according to their order in $\pi_1$.
\STATE For $1 \le k \le n$, let the job $i$ for which $\pi_2(i)=k$ perform a best-response move (migrate to $M_2$ if this reduces its completion time).
\end{algorithmic}
\end{algorithm}

\begin{theorem}
\label{thm:m2NE}
If $m=2$, then an NE exists and can be calculated efficiently.
\end{theorem}
\begin{proof}
Assume w.l.o.g. that $s_1=1$ and $s_2=s \le 1$. Consider Algorithm~\ref{NE2machines}, which initially assigns all the jobs on the fast machine. Then, the jobs are considered according to their order in $\pi_2$, and every job gets an opportunity to migrate to $M_2$.

Let us denote by $\widehat{\sigma}$ the schedule after the first step of the algorithm (where all the jobs are on $M_1$), and let $\sigma$ denote the schedule after the algorithm terminates. The following two claims show that after the termination of the algorithm, no job has a unilateral deviation that improves its cost, i.e., $\sigma$ is an NE.

\begin{claim}
\label{cl:dontleaveM1}
No job for which $\sigma_i=M_1$ has a beneficial migration.
\end{claim}
\begin{proof}
Assume by contradiction that job $i$ is assigned on $M_1$ and has a beneficial migration. Assume that $\pi_2(i)=k$. Job $i$ was offered to perform a migration in the $k$-th iteration of step 2 of the algorithm, but chose to remain on $M_1$. The only migrations that took place after the $k$-th iteration are from $M_1$ to $M_2$. Thus, if migrating is beneficial for $i$ after the algorithm completes, it should have been beneficial also during the algorithm, contradicting its choice to remain on $M_1$.
\end{proof}

\begin{claim}
\label{cl:dontleaveM2}
No job for which $\sigma_i=M_2$ has a beneficial migration.
\end{claim}
\begin{proof}
Assume by contradiction that the claim is false and let $i$ be the first job on $M_2$ (first with respect to $\pi_2$) that may benefit from returning to $M_1$.
Recall that, $\widehat{\sigma}$ denotes the schedule before job $i$ migrates to $M_2$ - during the second step of the algorithm.
Recall that $C_i(\sigma)$ is the completion time of job $i$ on $M_2$, and $C_i(\widehat{\sigma})$ is its completion time on $M_1$ before its migration.

Since the jobs are activated according to $\pi_2$ in the $2$-nd step of the algorithm, no jobs are added before job $i$ on $M_2$. Job $i$ may be interested in returning to $M_1$ only if some jobs that were processed before it on $M_1$, move to $M_2$ after its migration. Denote by $\Delta$ the set of these jobs, and let $\delta$ be their total processing time. Let $i'$ be the last job from $\Delta$ to complete its processing in $\sigma$. Since job $i'$ performs its migration out of $M_1$ after job $i$, and jobs do not join $M_1$ during step 2 of the algorithm, the completion time of $i'$ when it performs the migration is at most $C_{i'}(\widehat{\sigma})$. The migration from $M_1$ to $M_2$ is beneficial for $i'$, thus, $C_{i'}(\sigma) < C_{i'}(\widehat{\sigma})$. 

The jobs in $\Delta$ are all before job $i$ in $\pi_1$ and after job $i$ in $\pi_2$. 
Therefore, $C_{i'}(\widehat{\sigma}) < C_i(\widehat{\sigma})$, and $C_{i'}(\sigma) \ge C_{i}(\sigma)+\delta/s$.
Finally, we assume that $\sigma$ is not stable and $i$ would like to return to $M_1$. By returning, its completion time would be $C_i(\widehat{\sigma})-\delta$. Given that the migration is beneficial for $i$, and that $i$ is the first job who likes to return to $M_2$, we have that $C_i(\widehat{\sigma})-\delta < C_i(\sigma)$. 

Combining the above inequalities, we get 
\begin{align*}
 C_i(\widehat{\sigma}) &< C_i(\sigma)+\delta \le C_{i'}(\sigma)-(1/s-1)\delta \\
 &< C_{i'}(\widehat{\sigma})-(1/s-1)\delta < C_i(\widehat{\sigma})-(1/s-1)\delta,
\end{align*}
which contradicts the fact that $s \le 1$ and $\delta \ge 0$.
\end{proof}

By combining the claims~\ref{cl:dontleaveM1} and \ref{cl:dontleaveM2}, 
we conclude that no player has a beneficial deviation and $\sigma$ is an NE.
\end{proof}

The last class for which we show that an NE is guaranteed to exist is the class of games with at most $3$ jobs. Consider an instance consisting of $m$ machines with arbitrary priority lists, and $3$ jobs $a, b,\text{and}~c$. Let $M_1$ be a machine with the highest speed. Assume $\pi_1=(a,b,c)$. Clearly, job $a$ is on $M_1$ in every NE. An NE can be computed by adding jobs $b$ and $c$ greedily one after the other. If job $b$ picks $M_1$ then the resulting schedule is an NE. If job $b$ picks $M_2$ and job $c$ is then added before it on $M_2$, then job $b$ may migrate, to get a final NE. The above characterization is tight, as there exists a game with only $4$ jobs that has no NE: 

\begin{example}
\label{example:noNEsched4}
{\em
Consider the game $\hat{G}$ with $4$ jobs,  $N = \{a,b,c,d\}$, and three machines, $M = \{M_1, M_2, M_3\}$, with $\pi_1=(a,b,c,d)$, and $\pi_2= \pi_3= (d,b,c,a)$.  The speed of machine $j$ is $s_j=1/j$. 
The job processing times are $p_a = 5, p_b = 4, p_c =\frac{13}{3} \approx 4.33$ and $p_d = 9.25$.

Job $a$ is clearly on $M_1$ in every NE. Since $s_2>s_3$, Job $d$ is not on $M_3$ in any NE. We distinguish two different cases for job $d$.
\begin{enumerate}
    \item Job $d$ is on $M_1$. Then as job $b$ has the highest remaining priority among $b$ and $c$ on all machines, job $b$ picks the machine with the lowest completion time, which is $M_2$, and job $c$ lastly is then on $M_1$. As a result, $d$ prefers $M_2$ over $M_1$ (since $18.5 < 18.58$).
    \item Job $d$ is on $M_2$. Then as job $b$ has the highest remaining priority among $b$ and $c$ on all machines, job $b$ picks the machine with the lowest completion time, which is $M_1$, and job $c$ lastly is then on $M_3$ (since $13 < 13.33$). As a result, $d$ prefers $M_1$ (since $18.25< 18.5$) over $M_2$.
\end{enumerate}
Thus, the game $\hat{G}$ has no pure Nash equilibrium.}
\end{example}

A possible generalization of our setting considers unrelated machines, that is, for every job $i$ and machine $j$, $p_{ij}$ is the processing time of job $i$ if processed on machine $j$. We conclude this section with an example demonstrating that an NE need not exist in this environment already with only two unrelated machines.

\begin{example}
{\em 
Consider a game $G$ with 3 jobs, $N=\{a,b,c\}$, and two machines, $M=\{M_1,M_2\}$ with $\pi_1=(a,b,c)$ and $\pi_2=(c,a,b)$. The job processing times are $p_{a1}=5$, $p_{a2}=4$, $p_{b1}=7$, $p_{b2}=4$, $p_{c1}=1$ and $p_{c2}=7$. We show that $G$ has no NE. Specifically, we show that no assignment of job $c$ can be extended to a stable profile.

First, assume that job $c$ is on $M_1$. Then job $a$ has the highest remaining priority on the two machines and picks $M_2$. Given that job $a$ is on $M_2$, job $b$ prefers $M_1$ over $M_2$. However, job $c$ now prefers to pick $M_2$ as its completion time there is $7$, which is smaller than $8$.

Second, assume that job $c$ is on $M_2$. Then job $a$ has the highest remaining priority on the two machines and picks $M_1$. Given that job $a$ is on $M_1$, job $b$ prefers $M_2$. However, job $c$ now prefers to pick $M_1$ as its completion time there is $6$, which is smaller than $7$.
}
\end{example}
\subsection{Convergence of Best-Response Dynamics}
In this section we consider the question whether natural dynamics such as better-responses are guaranteed to converge to an NE. 
Given a strategy profile $\sigma$, a strategy $\sigma'_i$ for job $i\in N$ is a better-response if $C_i(\sigma'_i,\sigma_{-i})<C_i(\sigma)$.

We show that every sequence of best-response converge to an NE for every instance $G \in \G_1 \cup \G_2 \cup \G_3 \cup \G_4$. 

\begin{theorem} 
Let $G$ be a game instance in $\G_1 \cup \G_2 \cup \G_3 \cup \G_4$. Any best-response sequence  
in $G$ converges to an NE.
\end{theorem}
\begin{proof}
The proof has the same structure for all four classes. Assume that best-response dynamics (BRD) does not converge. Since the number of different profiles is finite, this implies that the sequence of profiles contains a loop. That is, the sequence includes a profile $\sigma_0$, starting from which jobs migrate and eventually return to their strategy in $\sigma_0$. Let $\Gamma$ denote the set of jobs that perform a migration during this loop. For each of the four classes we identify a job $i \in \Gamma$ such that once job $i$ migrates, it cannot have an additional beneficial move.

Consider first the case $G \in \G_1$, that is, a game with unit jobs.
Let $C_{min}$ be the lowest cost of a job in $\Gamma$ during the BR-cycle. Let $M_1$ be a machine on which $C_{min}$ is achieved. Let $i$ be the job achieving cost $C_{min}$ on $M_1$ with the highest priority on $M_1$ among the jobs in $\Gamma$. Once $i$ achieves cost $C_{min}$, its cost does not increase, as no job is added to $M_1$ before it. Job $i$ cannot have an additional beneficial move, as this will contradict the definition of $C_{min}$.


We turn to consider games in $\G_2$, that is, $G$ is played on two machines. W.l.o.g., assume $s_1=1$ and $s_2=s \le 1$.  Let $i$ be the job in $\Gamma$ with highest priority in $\pi_2$. 
Given that BRD loops and that $i \in \Gamma$, it holds that during the BR sequence $i$ migrates from $M_1$ to $M_2$ and then back from $M_2$ to $M_1$. 

We show that once $i$ moves from $M_1$ to $M_2$, moving back to $M_1$ cannot be beneficial for it.
Let $\sigma^\prime$ denote the schedule before job $i$ migrates from $M_1$ to $M_2$.
Assume by contradiction that $i$ may benefit from returning to $M_1$. 
Let $L_1$ be the total processing time of jobs on $M_1$ that precede $i$ on $\pi_1$ in $\sigma^{\prime}$. We have that $C_i(\sigma^{\prime})=L_1+p_i$.
Let $L_2$ be the the total processing time of jobs in $N \setminus \Gamma$ that precede $i$ on $\pi_2$.
Since $i$ has the highest priority among $\Gamma$ on $M_2$, its cost while on $M_2$ is $(L_2+p_i)/s$, independent of other jobs leaving and joining $M_2$.
The migration of $i$ from $M_1$ to $M_2$ is beneficial, thus, $L_1+p_i > (L_2+p_i)/s$.
Migrating back to $M_1$ may become beneficial only if the total processing time of job that would precede it on $M_1$ is less than $L_1$, thus, at least one job that precedes $i$ on $\pi_1$ migrates out of $M_1$ when $i$ is on $M_2$. Let $k$ be the last job, for which $\pi_1(k) < \pi_1(i)$ that have left $M_1$ when $i$ is on $M_2$. Following $k$'s migration the processing time of jobs on $M_1$ that precede $i$ in $\pi_1$ is $L'_1$. Migrating back is beneficial for $i$, thus, $L'_i+p_i <(L_2+p_i)/s$ (additional jobs may join $M_1$ after $k$ leaves it, but this only makes $M_1$ less attractive for $i$). 
Since $\pi_2(k) > \pi_2(i)$, the cost of $k$ after its migrating to $M_2$ is at least $(L_2+p_i+p_k)/s$. $k$'s migration from $M_1$ to $M_2$ is beneficial, thus, $L'_1+p_k > (L_2+p_i+p_k)/s$. 
By combining the above inequalities we reach a contradiction. Specifically, 
$L'_i+p_i < (L_2+p_i)/s =(L_2+p_i+p_k)/s-p_k/s < L'_i+p_k -p_k/s \le L'_i$.
We conclude that job $i$ cannot benefit from returning to $M_1$ and thus, cannot be involved in the BRD-cycle.


Assume next that $G \in \G_3$, that is, machines have identical speeds. Let $t$ be the lowest start time of a job in $\Gamma$ during the BR-cycle. Let $M_1$ be a machine on which $t$ is achieved. Let $i$ be the job in $\Gamma$ with highest priority on $M_1$. Clearly, once $i$ achieves start time, $t$, it cannot have an additional beneficial move, as this will contradict its choice.


Finally, if $G \in \G_4$, that is, when machines share a global priority list, then once the job in $\Gamma$ with the highest priority migrates, it selects the machine with the lowest total processing time of jobs in $N \setminus \Gamma$ that precedes it, and cannot have an additional beneficial move later.

\end{proof}

\section{Equilibrium Inefficiency}

Two common measures for evaluating the quality of a schedule are the makespan, 
given by $C_{max}(\sigma) = \max_{i \in N} C_i(\sigma)$, and the sum of completion times, given by $\sum_{i \in N} C_i(\sigma)$. In this section we analyze the equilibrium inefficiency with respect to each of the two objectives, for each of the four classes for which an NE is guaranteed to exist.

We begin with $\G_1$, the class of instances with unit jobs. For this class we show that allowing arbitrary priority lists does not hurt the social cost, even on machines with different speeds. 

\begin{theorem}
PoA$(\G_1) = $ PoS$(\G_1)= 1$ for both the min-makespan and the sum of completion times objective.
\end{theorem}

\begin{proof}
Let $\sigma$ be a schedule of unit jobs. The quality of $\sigma$ is characterized by the vector $(n_1(\sigma), n_2(\sigma),\ldots,n_m(\sigma))$ specifying the number of jobs on each machine. The makespan of $\sigma$ is given by $\max_j n_j(\sigma)/s_j$, and the sum of completion times in $\sigma$ is $\sum_j \frac{n_j(\sigma)(n_j(\sigma)+1)}{2s_j}$. 

Theorem \ref{thm:unitjobs} shows that assigning the jobs greedily, where on each step a job is added on a machine on which the cost of the next job is minimized, yields an NE. Let $\sigma^{\star}$ denote the resulting schedule, and let $C_1(\sigma^{\star}) \le C_2(\sigma^{\star}) \le \ldots \le C_n(\sigma^{\star})$ be the sorted vector of jobs' completion times in $\sigma^{\star}$. The proof proceeds by showing that this vector corresponds to schedules that minimize the makespan, as well as the sum of completion times. Also, we show that every NE schedule induces the same cost vector as $\sigma^{\star}$.

First, we show that $\sigma^{\star}$ achieves the minimum makespan. Assume that there exists a schedule $\sigma'$ such that $\max_j n_j(\sigma')/s_j < \max_j  n_j(\sigma^{\star})/s_j$. Let $M_1=\mbox{argmax}_j n_j(\sigma^{\star})/s_j$. It must be that $n_{M_1}(\sigma') < n_{M_1}(\sigma^{\star})$. Since $\sum_j n_j(\sigma')=\sum_j n_j(\sigma^{\star})=n$, there must be a machine $M_2$ such that $n_{M_2}(\sigma^{\star}) < n_{M_2}(\sigma')$. 
Thus, the last job on machine $M_1$ in $\sigma^{\star}$ can benefit from migrating to machine $M_2$, as its cost will be at most $(n_{M_2}(\sigma^{\star})+1)/s_{M_2} \le n_{M_2}(\sigma')/s_{M_2} \le \max_j n_j(\sigma')/s_j < \max_j n_j(\sigma^{\star})/s_j$. This contradicts the assumption that $\sigma^{\star}$ is an NE.

Second, we analyze the sum of completion times objective. For a schedule $\sigma$, the sum of completion times is $\sum_j  (1 + \ldots +n_j)/s_j$. Using similar arguments, if $\sigma^{\star}$ is not optimal with respect to the sum of completion times, there exists a beneficial migration from a machine whose contribution to the sum is maximal, to a machine with a lower contribution. 

Now, let $\sigma$ be an NE schedule with sorted cost vector and let $C_1(\sigma) \le C_2(\sigma) \le \ldots \le C_n(\sigma)$, and assume by contradiction that it has a different cost vector than $\sigma^*$. Let $i$ be the minimal index such that $C_i(\sigma^{\star}) \neq C_i(\sigma)$. Since $\sigma$ and $\sigma^{\star}$ agree on the costs of the first $i-1$ jobs, and since $\sigma^{\star}$ assigns the $i$-th job on a minimal-cost machine, it holds that $C_i(\sigma^{\star}) < C_i(\sigma)$. We get a contradiction to the stability of $\sigma$ - since some job can reduce its cost to $C_i(\sigma^{\star})$. The first and the second step concludes the proof of theorem. 
\end{proof}





In Theorem \ref{thm:m2NE} it is shown that an NE exists for any instance on two related machines. We now analyze the equilibrium inefficiency of this class. 
Let $\G_2^s$ denote the class of games played on two machines
with speeds $s_1=1$ and $s_2=s \le 1$.

\begin{theorem}
For the min-makespan objective, PoA$(\G_2^s) = $PoS$(\G_2^s) = s+1$ if $s \le \frac{\sqrt{5}-1}{2}$, and PoA$(\G_2^s) = $PoS$(\G_2^s) = \frac{s+2}{s+1}$ if $s > \frac{\sqrt{5}-1}{2}$.
\end{theorem}
\begin{proof}
Let $G \in \G_2^s$. Let $W=\sum_i p_i$ be the total processing time of all jobs.
Assume first that  $s \le \frac{\sqrt{5}-1}{2}$.
For the minimum makespan objective, $OPT(G) \ge W/(1+s)$.
Also, for any NE $\sigma$, we have that $C_{max}(\sigma) \le W$, since every job can migrate to be last on the fast machine and have completion time at most $W$. Thus, PoA $\le s+1$.

Assume next that $s>\frac{\sqrt{5}-1}{2}$. Let job $a$ be a last job to complete in a worst Nash equilibrium $\sigma$, $p_1$ be the total processing time of all jobs different from $a$ on machine $1$, and $p_2$ be the total processing time of all jobs different from $a$ on machine $2$ in $\sigma$. Then since $\sigma$ is a Nash equilibrium, $C_{max}(\sigma)\leq p_1+p_a$ and $C_{max}(\sigma)\leq (p_2+p_a)/s$. Combining these two inequalities yields
\[C_{max}(\sigma)\leq \frac{W+p_a}{1+s}\leq \frac{s+2}{s+1}\cdot OPT(G),\]
where for the inequality we use that $OPT(G) \ge W/(1+s)$ and $OPT(G) \ge p_a$, and thus PoA$\leq \frac{s+2}{s+1}$.

For the PoS lower bound, assume first that $s<\frac{2}{\sqrt{5
}+1}$. Consider an instance consisting of two jobs, $a$ and $b$, where $p_a=1$ and $p_b=1/s$. The priority lists are $\pi_1=\pi_2=(a,b)$. The unique NE is that both jobs are on the fast machine. $C_a(\sigma)=1, C_b(\sigma)=1+1/s$. For every $s < \frac{\sqrt{5}-1}{2}$, it holds that $1+1/s< 1/s^2 $, therefore, job $b$ does not have a beneficial migration. An optimal schedule assigns job $a$ on the slow machine, and both jobs complete at time $1/s$. The corresponding PoS is $s+1$. \footnote{For $s=\frac{\sqrt{5}-1}{2}$, by taking $p_b=1/s+\epsilon$, the PoS approaches $(s+2)/(s+1)$ as $\epsilon \rightarrow 0$.}

Assume now that $s > \frac{\sqrt{5}-1}{2}$. 
Consider an instance consisting of three jobs, $x$, $y$ and $z$, where $p_x=1, p_y= s^2+s-1$, and $p_z=s+1$.
The priority lists are $\pi_1=\pi_2=(x,y,z)$. 
Note that $p_y \ge 0$ for every $s \ge \frac{\sqrt{5}-1}{2}$. In all NE, job $x$ is on the fast machine, and job $y$ is on the slow machine. Indeed, job $y$ prefers being alone on the slow machine since $s^2+s>\frac{s^2+s-1}{s}$. Job $z$ is indifferent between joining $x$ on the fast machine or $y$ on the slow machine, since $1+p_z=(p_y+p_z)/s=s+2$. In an optimal schedule, job $z$ is alone on the fast machine, and jobs $x$ and $y$ are on the slow machine.
Both machines have the same completion time $s+1$.
The PoS is $\frac {s+2} {s+1}$. 
\end{proof}

\begin{theorem}
\label{thm:EqInef2}
For the sum of completion times objective, PoA$(\G_2^s) = \Theta(n)$ and PoS$(\G_2^s) = \Theta(n)$ for all $s \le 1$. 
\end{theorem}

\begin{proof}
For the upper bound, note that $OPT(G) \ge \sum_i p_i$ and in every NE schedule $\sigma$, $C_i(\sigma) \le \sum_i p_i$. This implies PoA $= \Theta(n)$.
For the PoS lower bound, consider an instance consisting of a set $Z$ of $n-2$ jobs with processing time $\epsilon$, and two jobs, $a$ and $b$, where $p_a=1$ and $p_b=s$.
The priority lists are $\pi_1=\pi_2=(a,b,Z)$. Note that $p_a+p_b > p_b/s$, therefore, in every NE, job $a$ is first on $M_1$ and job $b$ is first on $M_2$. Thus, every $\epsilon$-job has completion time at least $1$. The sum of completion times is at least $n+ O(n^2)\epsilon$. An optimal schedule assigns $a$ and $b$ on $M_1$ and all the $\epsilon$-jobs on $M_2$. The sum of completion times is at most $3+O(n^2)\epsilon/s$. For small enough $\epsilon$, we get that the PoS is $\Theta(n)$. 
\end{proof}

We turn to analyze the equilibrium inefficiency of the class $\G_3$, consisting of games played on identical-speed machines, having machine-based priority lists. The proof of the following theorem is based on the observation that every NE schedule is a possible outcome  of Graham's {\em List-scheduling} (LS) algorithm \cite{Gra66}. 
\begin{theorem}\label{thm:iden}
For the min-makespan objective, PoA$(\G_{3})=$PoS $(\G_{3})= 2-\frac 1 m$.
\end{theorem}

\begin{proof}
Let $\sigma$ be an NE schedule. We claim that $\sigma$ is a possible outcome of Graham's {\em List-scheduling} algorithm \cite{Gra66}. Indeed, assume that List-scheduling is performed and the jobs are considered according to their start time in $\sigma$. Every job selects its machine in $\sigma$, as otherwise, we get a contradiction to the stability of $\sigma$. Since List-scheduling provides a $2- \frac 1 m$ approximation to the makespan, we get the upper bound of the PoA.

For the lower bound, given $m>1$, the following is an instance for which PoS$=2-\frac 1 m$.
The instance consists of a single job with processing time $m$ and $m(m-1)$ unit jobs.
In all priority lists, the heavy job is last and the unit jobs are prioritized arbitrarily. It is easy to verify that in every NE profile the unit jobs are partitioned in a balanced way among the machines, and the heavy job is assigned as last on one of the machines. Thus, the completion time of the heavy job is $2m-1$. On the other hand, an optimal assignment assign the heavy job on a dedicated machine, and partition the unit job in a balanced way among the remaining $m-1$ machines. In this profile, all the machines have load $m$. The corresponding PoS is $\frac{2m-1}{m} = 2 - \frac 1 m$. 
\end{proof}  

\begin{theorem}
\label{thm:ineffc1}
For the sum of completion times objective, PoA$(\G_{3})\le \frac {n-1} m +1$, and for every $\epsilon >0$, PoS$(\G_3) \ge \frac{n}{m}-\epsilon$. 
\end{theorem}
\begin{proof}
For the upper bound of the PoA, note that, independent of the number of machines, the sum of completion times is at least $\sum_i p_i$. Also, for every job $a$, if $a$ is not assigned on any machine, then there exists a machine with load at most $\frac{\sum _{i\neq a} p_i}{m}$, therefore, in every NE profile, the completion time of job $a$ is at most $\frac{\sum _{i\neq a} p_i}{m}+p_a$. Summing this equation for all the jobs, we get that the sum of completion times of any NE is at most 
$$\frac{n\sum _i p_i - \sum _i p_i}{m}+ \sum_i p_i = \sum_i p_i \frac{n-1}{m}+1.$$
We conclude that the PoA is at most $\frac {n-1} m +1$. 

For the PoS lower bound, given $m$, let $\epsilon \rightarrow 0$, and consider an instance with $n=km$ jobs, out of which, $m$ jobs $j_1,\ldots,j_m$ have length $1$ and the other $(k-1)m$ jobs have length $\epsilon$. Assume that $\pi_i$ gives the highest priority to $j_i$ then to all the $\epsilon$-jobs, and then to the other $m-1$ unit jobs. 

In every NE, machine $i$ processes first the unit-job $j_i$, followed by $k-1$ $\epsilon$-jobs. Thus, every job has completion time at least $1$. The sum of completion times is  $n +O(mk^2)\epsilon$.
On the other hand, an optimal solution assigns on machine $i$ a set of $k-1$ $\epsilon$-jobs followed by one unit-job $j_k$ for $k \neq i$, resulting in a sum of completion times of $m+ O(mk^2)\epsilon$.
The PoS tends to $\frac{n}{m}$ as $\epsilon$ decreases.
\end{proof}


The last class of instances for which an NE is guaranteed to exist includes games with a global priority list, and is denoted by $\G_4$. It is easy to verify that for
this class, the only NE profiles are those produced by List-Scheduling algorithm, where the jobs are considered according to their order in the priority list. Different NE may be produced by different tie-breaking rules. Thus, the equilibrium inefficiency is identical to the approximation ratio of LS, as analyzed in \cite{10.1145/167088.167248}. The tie-breaking in an execution of LS can be controlled by adding some high priority jobs with very small processing times. Thus, since the analysis of LS is tight \cite{10.1145/167088.167248}, it is easy to show that the PoS is $\Theta(\log{m})$ as well.
\begin{theorem}
For the min-makespan objective, PoS$(\G_4)=$PoA$(\G_4)=\Theta(\log{m})$.
\end{theorem}

For the sum of completion times objective, we note that the proof of Theorem \ref{thm:EqInef2} for two related machines uses a global priority list. The analysis of the PoA is independent of the number and speeds of machines.
\begin{theorem}
For the sum of completion times objective,  PoA$(\G_4) = \Theta(n)$ and PoS$(\G_4) = \Theta(n)$.
\end{theorem}

\section{Hardness of Computing an NE with Low Social Cost}
Correa and Queyranne~\cite{correa2012efficiency} showed that if all the machines have the same speeds, but arbitrary priority lists, then an NE is guaranteed to exist, and can be calculated by a simple greedy algorithm.

In this section we discuss the complexity of computing a good NE in this setting. We refer to both objectives of minimum makespan and minimum sum of completion times. For both objectives, our results are negative. Specifically, not only that it is NP-hard to compute the best NE, but it is also hard to approximate it, and to compute an NE whose social cost is better than the one guaranteed by the PoA bound.

Starting with the minimum makespan, in Theorem \ref{thm:iden}, we have shown that the PoA for this objective is at most $2-\frac 1 m$. We show that we cannot hope for a better algorithm than the simple greedy algorithm. More formally, we prove that it is NP-hard to approximate the best NE within a factor of $2-\frac 1 m-\epsilon$ for all $\epsilon>0$.

\begin{theorem}
\label{thm:hardApproxCmax}
If for all machines $s_j=1$, then it is NP-hard to approximate the best NE w.r.t. the makespan objective within a factor of $2-\frac{1}{m}-\epsilon$ for all $\epsilon >0$.
\end{theorem}
\begin{proof}
We show that for every $\epsilon>0$, there is an instance on $m$ identical machines for which it is NP-hard to decide whether the game has an NE profile with makespan at most $m+3\epsilon$ or at least $2m-1$.

The hardness proof is by a reduction from 
$3$-bounded $3$-dimensional matching ($3$DM-$3$).
Recall that the input of the 3DM-3 problem is a set of triplets $T \subseteq X \times Y \times Z$, where $|T|\ge n$ and $|X|=|Y|=|Z|=n$, and the number of occurrences of every element of $X
\cup Y \cup Z$ in $T$ is at most $3$. The goal is to decide whether $T$ has a 3-dimensional matching of size $n$.

Given an instance of $3$DM-$3$ and $\epsilon >0$, consider the following game on $m=|T|+2$ machines, $M_1,M_2,\ldots,M_{|T|+2}$. 
The set of jobs includes job $a$ with processing time $m$, job $b$ with processing time $m-1$, a set $D$ of $|T|-n$ dummy jobs with processing time $3\epsilon$, two dummy jobs $d_1,d_2$ with processing time $2\epsilon$, a set $U$ of $(m-1)^2$ unit jobs, and $3n$ jobs with processing time $\epsilon$ - one for each element in $X \cup Y \cup Z$.

We turn to describe the priority lists. We remark that when the list includes a set, it means that the set elements appear in arbitrary order. The symbol $\phi$ means that the remaining jobs appear in arbitrary order. For the first machine, $\pi_1=(d_1,b,a,U,\phi)$. For the second machine, $\pi_2=(d_2,X,Y,Z,b,U,a,d_1)$.
The $m-2$ right machines are {\em triplet-machines}. 
For every $t=(x_i,y_j,z_k) \in T$, the priority list of the triplet-machine corresponding to $t$ is $(D,x_i,y_j,z_k,U,\phi)$.

The heart of the reduction lies in determining the priority lists. The idea is that if a $3$D-matching exists, then job $b$ would not prefer $M_1$ over $M_2$. This will enable job $a$ to be assigned early on $M_1$. However, if a $3$D-matching does not exist, then some job originated from the elements in $X \cup Y \cup Z$ will precede job $b$ on $M_2$, and $b$'s best-response would be on $M_1$. The jobs in $U$ have higher priority than job $a$ on all the machines except for $M_1$, thus, unless job $a$ is on $M_1$, it is assigned after $|U|/(m-1)$ unit-jobs from $U$, inducing a schedule with high makespan.

Observe that in any NE, the two dummy jobs with processing time $2\epsilon$ are assigned as the first jobs on $M_1$ and $M_2$. Also, the dummy jobs in $D$ have the highest priority on the triplet-machines, thus, in every NE, there are $|D|=|T|-n$ triplet-machines on which the first job is from $D$. 

Figure~\ref{fig:hardCmax} provides an example for $m=5$.

\begin{figure*}[ht]
\begin{center}
\includegraphics[height=6cm]{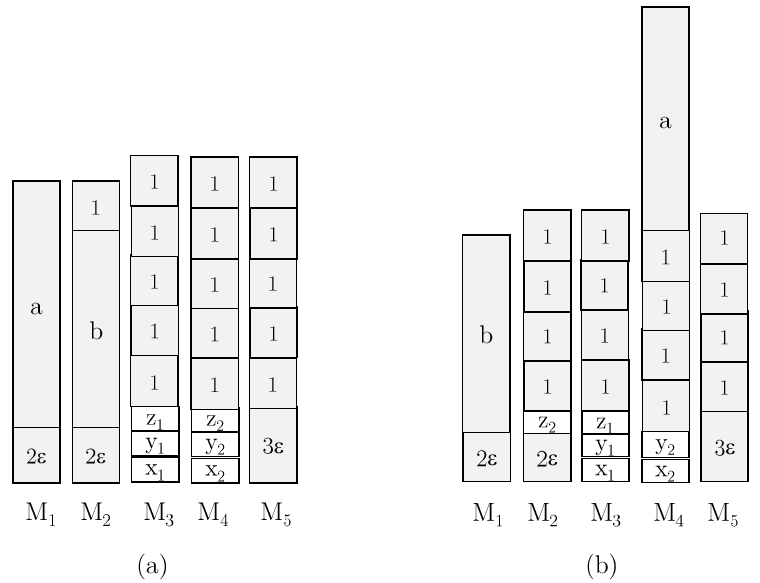}
\end{center}
\caption{Let $n=2$ and $T=\{(x_1,y_1,z_1),(x_2,y_2,z_2),(x_1,y_2,z_2)\}$.
 (a) an NE given the matching $T'=\{(x_1,y_1,z_1),(x_2,y_2,z_2)\}$. The makespan is $5+3\epsilon$. (b) an NE if a matching of size $2$ is not found. Job $z_2$ is stable on $M_2$, thus, job $b$ prefers $M_1$ over $M_2$. The makespan is $9+2\epsilon$.}
\label{fig:hardCmax}
\end{figure*}

In order to complete the gap reduction, we need the following two claims for the upper and lower threshold. First, we show that if a perfect matching exists, then this guarantees an NE in the associated scheduling problem instance with makespan at most $m+3\epsilon$.
\begin{claim}\label{clm:1:hardApproxCmax}
If a $3$D-matching of size $n$ exists, then there is an NE with makespan $m+3\epsilon$.
\end{claim}
\begin{proof}
Let $T'$ be a matching of size $n$. Assign the jobs of $X \cup Y \cup Z$ on the triplet-machines corresponding to $T'$ and the jobs of $D$ on the remaining triplet-machines. Also, assign $d_1$ and $d_2$ on $M_1$ and $M_2$ respectively. $M_1$ and $M_2$ now have load $2\epsilon$ while the triplet machines have load $3\epsilon$. Next, assign job $a$ on $M_1$ and job $b$ on $M_2$. Finally, add the unit-jobs as balanced as possible: $m$ jobs on each triplet-machine and a single job after job $b$ on $M_2$. It is easy to verify that the resulting assignment is an NE. Its makespan is $m+3\epsilon$.  
\end{proof}

The next claim proves the other direction of the reduction. That is, any NE with makespan less than $2m-1$ induces a perfect matching.
\begin{claim}\label{clm:2:hardApproxCmax}
If there is an NE with makespan less than $2m-1$, then there exists a $3$D-matching of size $n$.
\end{claim}
\begin{proof}
Let $\sigma$ be an NE whose makespan is less than $2m-1$. Since $p_a=m$ and $p_b=m-1$, this implies that $a$ is not assigned after $b$ on $M_1$ or on $M_2$. Also, since jobs of $U$ have higher priority than $a$ on all the machines except for $M_1$, it holds that $a$ is not assigned after $m-1$ unit-jobs. Thus, it must be that job $a$ is processed on $M_1$ and job $b$ is not on $M_1$. Job $b$ does not prefer $M_1$ over $M_2$ only if it starts its processing right after job $d_2$ on $M_2$. Since the jobs of $X \cup Y \cup Z$ have higher priority than job $b$ on $M_2$, they are all assigned on triplet-machines and starts their processing after jobs of total processing time at most $2\epsilon$. Thus, every triplet machine processes at most three jobs of $X \cup Y \cup Z$ - the jobs corresponding to the triplet, whose priority is higher than the priority of the unit-jobs of $U$. Moreover, since the jobs of $D$ have higher priority on the triplet-machines, there are $|T|-n$ triplet-machines on which the jobs of $D$ are first, and exactly $n$ machines each processing exactly the three jobs corresponding to the machine's triplet. Thus, the assignment of the jobs from $X \cup Y \cup Z$ on the triplet-machines induces a matching of size $n$.
\end{proof}

Claims \ref{clm:1:hardApproxCmax} and \ref{clm:2:hardApproxCmax} conclude the proof of the theorem. 
\end{proof}

We turn to analyze the complexity of computing the best NE with respect to the sum of completion times. Traditionally, this objective is simpler than minimizing the makespan, as the problem can be solved efficiently by SPT-rule if there are no priorities. We show that even in the simple case of identical machines, in which an NE is guaranteed to exists~\cite{correa2012efficiency}, it is NP-hard to approximate the solution's value. 

\begin{theorem}
\label{thm:hardApproxSumC}
If for all machines $s_j=1$, then, for any $r>1$, it is NP-hard to approximate the best NE w.r.t. the sum of completion times within factor $r$.
\end{theorem}
\begin{proof}
Given $m,r$, let $k$ be a large integer such that $\frac{m+k}{m+1} >r$. 
We show that for every $r>1$, there is an instance on $m$ identical machines for which it is NP-hard to decide whether the game has an NE profile with sum of completion times at most $m+1$ or more than $m+k$.

The hardness proof is, again, by a reduction from 
$3$-bounded $3$-dimensional matching ($3$DM-$3$).
Given an instance of $3$DM-$3$ and $r>1$, consider the following game on $m=|T|+2$ machines, $M_1,M_2,\ldots,M_{|T|+2}$. 
Recall that $k$ satisfies $\frac{m+k}{m+1} >r$. Let $\epsilon>0$ be a small constant, such that $(k^2+3k+6m)\epsilon < 1$. 
The set of jobs includes job $a$ with processing time $\epsilon$, job $b$ with processing time $1$, a set $D$ of $|T|-n$ dummy jobs with processing time $3\epsilon$, two dummy jobs $d_1,d_2$ with processing time $2\epsilon$, a set $U$ of $m-1$ unit jobs, $3n$ jobs with processing time $\epsilon$ - one for each element in $X \cup Y \cup Z$, and a set $K$ of $k$ jobs with processing time $\epsilon$. Note that there are exactly $m$ unit jobs (the job $b$ and the jobs of $U$), while all other jobs have $O(\epsilon)$ processing time.

We turn to describe the priority lists. Note that, when the list includes a set, it means that the set elements appear in arbitrary order. The symbol $\phi$ means that the remaining jobs appear in arbitrary order. For the first machine, $\pi_1=(d_1,b,a,K,U,\phi)$. 
For the second machine, $\pi_2=(d_2,X,Y,Z,b,U,a,K,\phi)$.
The $m-2$ right machines are {\em triplet-machines}. 
For every $t=(x_i,y_j,z_k) \in T$, the priority list of the triplet-machine corresponding to $t$ is $(D,x_i,y_j,z_k,a,U,K,\phi)$. 

The heart of the reduction lies in determining the priority lists. The idea is that if a $3$D-matching exists, then job $b$ would not prefer $M_1$ over $M_2$. This will enable job $a$ and all the tiny jobs of $K$ to be assigned early on $M_1$ each having completion time at most $(k+3)\epsilon$. However, if a $3$D-matching does not exist, then some job originated from the elements in $X \cup Y \cup Z$ will precede job $b$ on $M_2$, and $b$'s best-response would be on $M_1$. The jobs in $U$ have higher priority than $a$ and $K$ on $M_2$, thus, on every machine there would be at least one job of length $1$ that precedes the jobs of $K$, implying that the sum of completion times will be more than $m+k$.

Observe that in any NE, the two dummy jobs with processing time $2\epsilon$ are assigned as the first jobs on $M_1$ and $M_2$. Also, the dummy jobs in $D$ have the highest priority on the triplet-machines, thus, in every NE, there are $|D|=|T|-n$ triplet-machines on which the first job is from $D$. 

Figure~\ref{fig:hardCsum} provides an example for $m=5$.

\begin{figure*}[ht]
\begin{center}
\includegraphics[height=4.3cm]{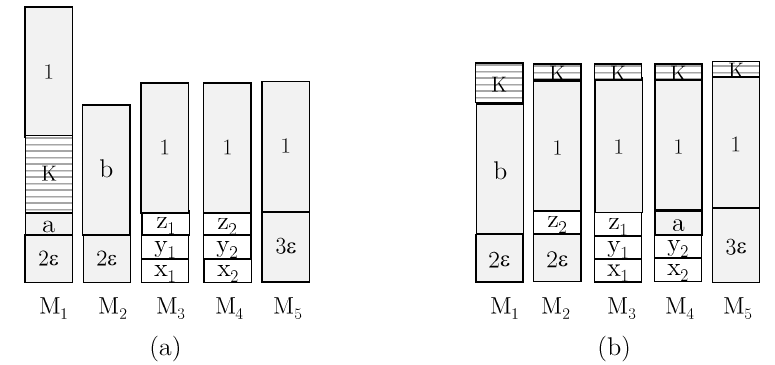}
\end{center}
\caption{Let $n=2$ and $T=\{(x_1,y_1,z_1),(x_2,y_2,z_2),(x_1,y_2,z_2)\}$.
 (a) an NE given the matching $T'=\{(x_1,y_1,z_1),(x_2,y_2,z_2)\}$. The jobs of $K$ start their processing at time $3\epsilon$. (b) an NE if a matching of size $2$ does not exist. Job $z_2$ selects $M_2$, thus, job $b$ prefers $M_1$ over $M_2$. The jobs of $K$ start their processing at time at least $1+2\epsilon$.}
\label{fig:hardCsum}
\end{figure*}

The following claims prove the lower and the upper threshold in the gap instance of the scheduling problem.
\begin{claim}\label{clm:1:hardApproxSumC}
If a $3$D-matching of size $n$ exists, then there is an NE with sum of completion time at most $m + 1$.
\end{claim}
\begin{proof}
Let $T'$ be a matching of size $n$. Assign the jobs of $X \cup Y \cup Z$ on the triplet-machines corresponding to $T'$ and the jobs of $D$ on the remaining triplet-machines. Also, assign $d_1$ and $d_2$ on $M_1$ and $M_2$ respectively. $M_1$ and $M_2$ now have load $2\epsilon$ while the triplet machines have load $3\epsilon$. Next, assign job $a$ and the jobs of $K$ on $M_1$, and job $b$ on $M_2$. Finally, add one unit-job on each triplet-machine and on $M_1$. It is easy to verify that the resulting assignment is an NE. The jobs of $K$ are not delayed by unit jobs, so each of them completes at time at most $(k+3)\epsilon$. The other jobs with processing time $O(\epsilon)$ contributes at most $6\epsilon$ to the sum of completion times on every machine, and every unit job completes at time $1+O(\epsilon)$. Thus, the sum of completion times is $m+f(\epsilon)$, where $\epsilon$ was chosen such that $f(\epsilon)<1$.
\end{proof}

\begin{claim}\label{clm:2:hardApproxSumC}
If there is an NE with sum of completion times less than $m+k$, then there exists a $3$D-matching of size $n$.
\end{claim}
\begin{proof}
Let $\sigma$ be an NE whose sum of completion times is less than $m+k$. 
There are $m$ unit jobs, and in any NE, each is processed on a different machine, as otherwise, some machine has load $f(\epsilon)$, and the second unit job on a machine has a beneficial migration. 
In order to have sum of completion times less than $m+k$, at least one job from $K$ is not assigned after a unit job.  The only machine on which jobs from $K$ may precede a unit job is $M_1$, where jobs of $K$ may precede a job from $U$. This is possible only if job $b$ is not processed on $M_1$. Job $b$ does not prefer $M_1$ over $M_2$ only if it starts its processing right after job $d_2$ on $M_2$. Since the jobs of $X \cup Y \cup Z$ have higher priority than job $b$ on $M_2$, they are all assigned on triplet-machines and starts their processing after jobs of total processing time at most $2\epsilon$. Thus, every triplet machine processes at most three jobs of $X \cup Y \cup Z$ - the jobs corresponding to the triplet, whose priority is higher than the priority of the unit-jobs of $U$. Moreover, since the jobs of $D$ have higher priority on the triplet-machines, there are $|T|-n$ triplet-machines on which the jobs of $D$ are first, and exactly $n$ machines each processing exactly the three jobs corresponding to the machine's triplet. Thus, the assignment of the jobs from $X \cup Y \cup Z$ on the triplet-machines induces a matching of size $n$.
\end{proof}

Claims \ref{clm:1:hardApproxSumC} and \ref{clm:2:hardApproxSumC} conclude the proof of theorem.
\end{proof}

Note that, given $m,r$, the game built in the reduction has $n < (r+3)m$ jobs. That is, $r > \frac{n}{m} -3$. Also, as shown in Theorem \ref{thm:ineffc1}, for the sum of completion times objective, PoA$(\G_{3})\le \frac {n-1} m +1$. Thus, the above analysis shows that up to a small additive constant, it is NP-hard the compute an NE that approximates the optimal sum of completion time better than the PoA.

\section{Conclusion and Open Problems}
Traditional analysis of coordination mechanisms assumes that jobs assigned to some machine are processed according to some policy, such as shortest or longest processing time. In this paper we explored the effect of having a different policy, given by an arbitrary priority list, for every machine. We showed that in general, an NE schedule may not exist, and it is NP-hard to identify whether a given game has an NE. On the other hand, for several important classes of instances, we showed that an NE exists and can be computed efficiently, and we bounded the equilibrium inefficiency with respect to the common measures of minimum makespan and sum of completion times. We also showed that natural dynamics converge to an NE for all these classes. 
In terms of computational complexity, we proved that even for the simple class of identical machines, for which an NE can be computed efficiently, it is NP-hard to compute an NE whose quality is better than the quality of the worst NE.

Our work leaves open several interesting directions for future work.
\begin{itemize}
    \item To the best of our knowledge, the problem of computing the social optimum of an instance is a new variant of scheduling with precedence constraints, that has not been studied before. The main difference from classical scheduling with precedence constraints is that 
    a priority list determine the scheduling priority for jobs on a specific machine, rather than for the entire schedule.
    Therefore, it is not possible to adopt known ideas and techniques.
    \item Since our game may not have an NE, it is natural to consider weaker notions of stability. In particular, for a parameter $\alpha \ge 1$, a profile is an $\alpha$-approximate NE if no job can change strategy such that the cost reduces by factor at least $\alpha$~\cite{Tim_Eva}. The existence and calculation of approximate NE profiles is still open. 
    \item A natural generalization is to consider games in which jobs have an arbitrary strategy space, and the cost of a job is the sum of the cost for the resources used, where each resource has its own priority list.
        
    
\end{itemize}

%
%
\bibliographystyle{abbrv}
\bibliography{bib}

\end{document}